\documentclass[twocolumn,amsthm]{autart}    

\usepackage{graphicx}          

\usepackage{amsmath,amssymb, amsfonts, mathrsfs}
\usepackage{graphicx}
\usepackage{xcolor}
\usepackage{comment}
\usepackage{dsfont}
\usepackage{csquotes}
\usepackage{algorithm}
\usepackage{algpseudocode}
\usepackage{caption, subcaption}
\usepackage{enumitem}
\usepackage{siunitx}
\usepackage{dsfont}
\usepackage{float}
\usepackage{dsfont}

\usepackage{algorithm} 

\newtheorem{theorem}{Theorem}[section]
\newtheorem{corollary}{Corollary}[theorem]

\newtheorem{remark}[theorem]{Remark}
\newtheorem{assumption}[theorem]{Assumption}

\newcommand{\ie}{\textit{i.e., }}
\newcommand{\eg}{\textit{e.g., }}

\DeclareMathOperator*{\argmin}{arg\,min}

\newcommand{\model}{h}
\newcommand{\Var}{\operatorname{Var}}
\newcommand{\Cov}{\operatorname{Cov}}

\begin{document}

\begin{frontmatter}

\title{Multi-Level Multi-Fidelity Methods for Path Integral and Safe Control} 

\thanks[footnoteinfo]{Corresponding author Yorie Nakahira.}

\author[CMU]{Zhuoyuan Wang}\ead{zhuoyuaw@andrew.cmu.edu},    
\author[Purdue]{Takashi Tanaka}\ead{tanaka16@purdue.edu},               
\author[GT]{Yongxin Chen}\ead{yongchen@gatech.edu},  
\author[CMU]{Yorie Nakahira}\ead{ynakahir@andrew.cmu.edu}

\address[CMU]{Carnegie Mellon University, Pittsburgh, PA, USA}  
\address[Purdue]{Purdue University, West Lafayette, IN, USA}             
\address[GT]{Georgia Institute of Technology, Atlanta, GA, USA}        


\begin{abstract}                          
Sampling-based approaches are widely used in systems without analytic models to estimate risk or find optimal control. However, gathering sufficient data in such scenarios can be prohibitively costly. On the other hand, in many situations, low-fidelity models or simulators are available from which samples can be obtained at low cost. In this paper, we propose an efficient approach for risk quantification and path integral control that leverages such data from multiple models with heterogeneous sampling costs. A key technical novelty of our approach is the integration of Multi-level Monte Carlo (MLMC) and Multi-fidelity Monte Carlo (MFMC) that enable data from different time and state representations (system models) to be jointly used to reduce variance and improve sampling efficiency. We also provide theoretical analysis of the proposed method and show that our estimator is unbiased and consistent under mild conditions. Finally, we demonstrate via numerical simulation that the proposed method has improved computation (sampling costs) vs. accuracy trade-offs for risk quantification and path integral control. 
\end{abstract}

\end{frontmatter}

\section{Introduction}

Optimal and safe control are key to stochastic autonomous systems. For example, industrial robots need to operate efficiently for assembly tasks while maintaining safe distances to humans, under uncertain and complex working environments. Existing stochastic optimal and safe control methods usually require estimation of cost-to-go functions~\cite{fleming2012deterministic} and safety probabilities~\cite{wang2022myopically}. However, to acquire accurate estimate of such values is not trivial, because huge numbers of samples are needed to capture rare failure events and it is often hard to directly sample from the physical systems. 
In reality however, people often have access to multiple models of the system with different computation cost. For example, in robotic control one might have simulation models with minimal computation cost and the actual physical system with considerable data acquisition cost.
In this work, we leverage different models of the system and different time discretizations of the dynamics to efficiently estimate the safety probability and optimal control.

For safety probability estimation, the problem reduces to estimating the probability of rare risk events as they are complementary~\cite{kumamoto2007satisfying}.
Standard Monte Carlo (MC) method can empirically estimate risk probabilities but with heavy computation~\cite{rubino2009rare,hu2023simplified}. To improve sample complexity, importance sampling~\cite{cerou2012sequential,botev2013markov} and subset simulation~\cite{au2001estimation} are used on top of standard MC. Multi-level Monte Carlo (MLMC)~\cite{xu2020efficient} and multi-fidelity Monte Carlo (MFMC)~\cite{patsialis2021multi} are also used for risk quantification by simulating the system with different levels of time discretizations and different surrogate models, respectively. Both methods alone can achieve better sample complexity over standard MC, but a potential joint layering of both time discretization and system models has not been explored in the context of risk quantification. Besides, existing methods usually assume a static system model, while in this work we consider dynamical systems which is a more generic setting.

Path integral control is a type of numerical methods to solve stochastic optimal control problems by repeatedly performing forward MC rollouts of a sampled dynamics~\cite{kappen2005path,satoh2016iterative,zhang2021path}.
While time domain layering is considered in path integral control settings with MLMC~\cite{shumway2006path}, a potential combination with multi-fidelity models and low-dimensional representations have not been studied.
Besides, the variance and thus the sample complexity of the path integral scheme can be highly reduced if a near-optimal controller is sampled~\cite{thijssen2015path}. To acquire the form of such controller usually requires iterative sampling~\cite{pan2015sample}, but existing work did not leverage information from previous samplings to form more efficient estimators.

In this work, we propose a multi-level multi-fidelity (MLMF) framework to efficiently estimate the safety probability and optimal control in stochastic systems. Specifically, we extend the MLMC method to cases with different levels of models, as opposed to just different time discretizations. We then show that under certain conditions the MLMC method in state space is equivalent to the MFMC method. Based on the observation, we propose a unified multi-level multi-fidelity (MLMF) framework that incorporates layering in both time and state space, and derive unbiasedness property and optimal variance to show reduced sample complexities.
We also present how the proposed MLMF framework can estimate safety probability and path integral optimal control with potential low-dimensional representation of the state on a feature space, and show experiment results to validate its efficacy.

The proposed method has the following merits:
\begin{enumerate}
    \item Layering structure in both time and state space for more flexibility of usage and possibility of computation reduction (Alg.~\ref{alg:mlmf_app}).
    
    
    \item Unbiased and consistent estimation under mild conditions (Theorem~\ref{thm:unbiased_estimation} and Theorem~\ref{thm:optimal_variance}).

    \item Better sample complexity for estimating rare events such as risk probability (Fig.~\ref{fig:risk_comparison}).

    \item Layered and incremental sampling of near-optimal policy for efficient path integral control (Fig.~\ref{fig:cost_computation} and Fig.~\ref{fig:control_input}). 


\end{enumerate}

The rest of the paper is organized as follow. We present related work in Sec.~\ref{sec:related_work} and introduce the problem formulation in Sec.~\ref{sec:problem_formulation}. After that we present our proposed multi-level multi-fidelity framework in Sec.~\ref{sec:mlmf} and show theoretical analysis on unbiasedness and consistency of the estimation as well as optimal choice of the parameters in Sec.~\ref{sec:theoretical_analysis}.
We then demonstrate the usage of the proposed framework in safety probability and optimal control estimation in Sec.~\ref{sec:mlmf_application}, and present experiment results in Sec.~\ref{sec:experiments}. Finally, we conclude the paper in Sec.~\ref{sec:conclusion}.

\section{Related Work}
\label{sec:related_work}

\subsection{Multi-Fidelity Monte Carlo}
Multi-fidelity Monte Carlo (MFMC)~\cite{peherstorfer2016optimal} leverages different levels of system models, and construct a layered estimator to improve sample complexity over standard MC when different models with different computation cost are accessible.
MFMC has been used in uncertainty quantification~\cite{motamed2020multi,peherstorfer2018multifidelity}, risk assessment~\cite{patsialis2021multi}, micro-turbulence analysis~\cite{konrad2022data}, \textit{etc}.
However, to our best knowledge, no work has considered MFMC for path integral control. Besides, while some literature considers dimensionality reduction for static systems~\cite{fernandez2016review}, existing MFMC methods for dynamical control systems often consider surrogate models of the same dimension of the ground truth model~\cite{lau2023multi}. Our work explores the possibility of using dynamic models on a low-dimensional feature space.

\subsection{Multi-Level Monte Carlo}
Multi-level Monte Carlo (MLMC)~\cite{giles2015multilevel} incorporates multiple levels of estimate to form efficient estimator by increasing coupling and reducing variance between levels. This method alone can also improve sample complexity over standard MC, and has been used to solve stochastic differential equations (SDEs)~\cite{giles2014antithetic}, the exit times of SDEs~\cite{giles2018multilevel}, partial differential equations~\cite{abdulle2013multilevel}, and path integral control~\cite{jasra2022multilevel}, \textit{etc}. However, for dynamical systems, usually different levels of time/grid discretizations are used~\cite{abdulle2013stabilized}. Our work extends MLMC to incorporate different levels of state-space representations.

\section{Problem Formulation}
\label{sec:problem_formulation}
Consider the following system dynamics characterized by stochastic differential equations (SDEs)
\begin{equation}
\label{eq:dynamics}
    dx_t = f(x_t) \; dt + \sigma(x_t) (u_t \; dt + dw_t)
\end{equation}
where $x \in \mathbb{R}^n$ is the state, $u \in \mathbb{R}^m$ is the control input, $w$ is an $m$-dimensional standard Wiener process with $w_0 = \boldsymbol{0}$, and $\sigma: \mathbb{R}^n \rightarrow \mathbb{R}^m$ is a mapping from state to control space. 

For \textbf{safe control}, we assume the safe set is given by
\begin{equation}
\label{eq:safe_set_def}
    \mathcal{C}:=\{x: \phi(x) \geq 0\},
\end{equation}
where $\phi: \mathbb{R}^n \rightarrow \mathbb{R}$ is a differentiable function. 
The long-term safety probability starting from $x$ for time horizon $T$ is defined as
\begin{equation}
\label{eq:safe_prob_def}
    F(x) := \mathbb{P}(x_\tau \in \mathcal{C}, u_\tau = \mathcal{N}(x_\tau), \forall \tau \in [t, t+T] \mid x_t = x),
\end{equation}
where $\mathcal{N}: \mathbb{R}^n \rightarrow \mathbb{R}^m$ is a pre-defined control policy.
The goal is to estimate $F(x)$. 

To estimate the safety probability, one can run MC simulation of the system dynamics~\eqref{eq:dynamics} with the nominal controller $\mathcal{N}$ multiple times, and find the empirical ratio of safe trajectories to estimate the value of interest. Let $x^n$ denote the $n$-th trajectory, the estimation $P_{\text{safe}}$ is given by
\begin{equation}
\label{eq:safe_prob_est}
\begin{aligned}
    P_{\text{safe}} = \frac{1}{N} \sum \mathds{1}(\phi(x_{t}^n) > 0, \forall t \in [0, T]),
\end{aligned}
\end{equation}
where $N$ is the number of trajectories and $\mathds{1}(\cdot)$ is the indicator function.

For \textbf{optimal control}, we aim to find the control input that minimizes the following cost function
\begin{equation}
    J(u)=\mathbb{E}\left\{\int_0^T g\left(x_t, u_t\right) d t+\Psi\left(x_T\right)\right\},
\end{equation}
where $T$ is the horizon of the problem, $\Psi: \mathbb{R}^n \rightarrow \mathbb{R}$ is the terminal cost, $g$ is the running cost with the form
\begin{equation}
    g(x, u) = \ell(x) + \frac{1}{2}\|u\|^2,
\end{equation}
where $\ell(x)$ is the cost on states, and quadratic cost on control is assumed as in standard path integral control literature~\cite{thijssen2015path,theodorou2012relative,kappen2005path}.
The goal is to find the optimal control
\begin{equation}
\label{eq:optimal_control_def}
    u^* = \argmin_u \; J(u).
\end{equation}





For path integral control, one can run MC simulations of the system dynamics~\eqref{eq:dynamics} with a given controller \(u\) multiple times, and then compute a weighted average to extract the optimal control. The key idea is to recast the stochastic optimal control problem into a statistical inference problem, where trajectory costs act as exponential weights. Concretely, the optimal control can be estimated via 
\vspace{-22pt}
\begingroup
\setlength{\abovedisplayskip}{-15pt}
\setlength{\abovedisplayshortskip}{0pt}
\small
\begin{equation}
\label{eq:path_integral_control}
\begin{aligned}
    & P_{\text{control}} = u_t + \\
    & \lim _{s \rightarrow t} \frac{\sum\left\{\exp \left[-\int_t^T g\!\left(x_\tau^n, u_\tau^n\right) d \tau-\Psi\!\left(x_T^n\right)\right] \int_t^s d w_\tau^n \;\bigm|\; x_t^n=x\right\}}{(s-t) \sum\left\{\exp \left[-\int_t^T g\!\left(x_\tau^n, u_\tau^n\right) d \tau-\Psi\!\left(x_T^n\right)\right] \;\bigm|\; x_t^n=x\right\}},
\end{aligned}
\normalsize
\end{equation}
\endgroup
where the sum \(\sum\) is taken over \(N\) simulated trajectories, and \(x^n, u^n, w^n\) are realizations of the state, control, and noise for the \(n\)-th trajectory.  
Intuitively, each trajectory is assigned a weight $\exp[-\int_t^T g(x_\tau^n, u_\tau^n)\,d\tau - \Psi(x_T^n)]$, 
which encodes both the running cost \(g(\cdot)\) along the path and the terminal cost \(\Psi(\cdot)\). Trajectories that yield lower cost receive higher weight, while high-cost trajectories are exponentially suppressed. The numerator in~\eqref{eq:path_integral_control} collects weighted fluctuations of the noise increments \(\int_t^s dw_\tau^n\), and the denominator normalizes these contributions to ensure a proper probability-weighted average. We refer readers to~\cite{kappen2005path} for more detailed discussions on path integral control.

Now for \textbf{both problems}, the goal is to efficiently estimate $P_{\text{safe}}$ in~\eqref{eq:safe_prob_est} and $P_{\text{control}}$ in~\eqref{eq:path_integral_control}, for safety probability and path integral control estimation respectively. 
In this work, we consider cases where we have access to multiple models of the ground truth system dynamics~\eqref{eq:dynamics} of possibly different time discretizations, 
with different levels of accuracy and different sampling cost. 
In the following sections, we will introduce a scheme to effectively estimate the safety probability and optimal control by leveraging the access to multiple models and time discretizations.

\section{Multi-level Multi-fidelity Monte Carlo}
\label{sec:mlmf}
In this section, we introduce our proposed multi-level multi-fidelity (MLMF) Monte Carlo framework. We start by introducing multi-level Monte Carlo (MLMC). And then we introduce multi-fidelity Monte Carlo (MFMC) and derive its relationship with MLMC. Finally, we propose the unified MLMF framework.

\subsection{Multi-Level Monte Carlo}
We start by introducing standard MLMC. 
Assume there is a function $\model: \mathbb{R}^n \rightarrow \mathbb{R}$ and a random variable $Z$. 
The goal is to estimate $\model(Z)$. In safety probability estimation, $\model$ is $F$ in~\eqref{eq:safe_prob_def} with dynamics~\eqref{eq:dynamics} and $Z$ is the long-term trajectory $\{x_\tau, \tau \in [t, t+T]\}$. In path integral control, $\model$ is the optimal control transformation~\eqref{eq:optimal_control_def} with dynamics~\eqref{eq:dynamics} and $Z$ is the current state $x$.

Denote $P$ the value of interest to be estimated, and $P_1, P_2, \cdots, P_L$ are estimates from level $1$ to $L$ with increasing accuracy but also increasing cost. The MLMC estimation is given by
\begin{equation}
\label{eq:MLMC_estimator}
    \mathbb{E}[P] = \mathbb{E}[P_1] + \sum_{l=2}^L \mathbb{E}[P_l - P_{l-1}],
\end{equation}
and numerically we can get an estimate through
\begin{equation}
    Y_{\text{MLMC}} = N_1^{-1} \sum_{n=1}^{N_1}P_1^n + \sum_{l=2}^L \left\{N_l^{-1} \sum_{n=1}^{N_l}(P_l^n - P_{l-1}^n)\right\},
\end{equation}
where $N_l$ is the number of samples at level $l$, and $P_l^{n}$ is the estimate of $P_l$ with the $n$-th sample. Note that the same noise realizations are reused among all layers, so that $P_l$ are coupled for reduced estimation variance.


Assume the sampling cost for each layer $l$ is given by $C_l$, and the variance of the estimate at each layer is given by $V_l$. It is shown that for mean square error $\epsilon^2$, the sample complexity of MLMC is bounded by $O(\epsilon^{-2})$ when $C_l$ is exponentially increasing and $V_l$ is exponentially decreasing with regard to level $l$~\cite[Theorem 1]{giles2015multilevel}.

In previous literature, MLMC is used in time domain with different levels of discretization $\Delta t_l$ for each level $l$.
However, to our best knowledge, no work has considered MLMC with different levels of state representation. In our proposed framework to be introduced later, we will incorporate different state representations into the MLMC scheme.


\subsection{Multi-Fidelity Monte Carlo}

In multi-fidelity Monte Carlo (MFMC), we have a high-fidelity model $\model_L$, and lower-fidelity surrogate models $\model_1, \ldots, \model_{L-1}$. The high-fidelity model $\model_L$ is the ground truth model.\footnote{In MFMC literature (\eg\cite{peherstorfer2018multifidelity}) it is usually assumed that $\model_1$ is the high-fidelity model. Here we reverse the order to match with the formulation of MLMC.}
Assume the cost of evaluating a model $\model_l$ is $C_l$ for $l=1, \ldots, L$.
One evaluates model $\model_l$ for $N_l$ times and we have $N_1 \geq N_2 \geq \cdots \geq N_L > 0$.\footnote{This conditions ensures the first $N_l$ samples can be reused for level $l$ in~\eqref{eq:mfmc_estimator}, similar as in~\cite{peherstorfer2018multifidelity}.} Denote $\boldsymbol{z}_1, \ldots, \boldsymbol{z}_{N_l}$ the $N_l$ realizations of a random variable $Z$, and evaluate model $\model_l$ at the realizations to obtain the outputs $\model_l\left(\boldsymbol{z}_1\right), \ldots, \model_l\left(\boldsymbol{z}_{N_l}\right)$
for $l=1, \ldots, L$. Define the Monte Carlo estimators
\begin{equation}
\label{eq:mc_estimators_mf}
\begin{aligned}
    {Y}^{(N_l)}_l=\frac{1}{N_l} \sum_{j=1}^{N_l} \model_l\left(\boldsymbol{z}_j\right), & \quad
    {Y}^{(N_{l+1})}_l=\frac{1}{N_{l+1}} \sum_{j=1}^{N_{l+1}} \model_l\left(\boldsymbol{z}_j\right), \\
    {Y}^{(N_L)} & = \frac{1}{N_L} \sum_{j=1}^{N_L} \model_L\left(\boldsymbol{z}_j\right),
\end{aligned}
\end{equation}
for $l=1, \ldots, L-1$.
Note that the first $N_{l+1}$ samples are reused for ${Y}^{(N_l)}_l$ and ${Y}^{(N_{l+1})}_l$, which enables coupling and is the key to variance reduction.
The MFMC method combines the MC estimators~\eqref{eq:mc_estimators_mf} into the following estimator
\begin{equation}
\label{eq:mfmc_estimator}
    Y_{\text{MFMC}}={Y}^{(N_L)}+\sum_{l=1}^{L-1} a_l\left({Y}^{(N_l)}_l-{Y}^{(N_{l+1})}_l\right).
\end{equation}
The MFMC estimator~\eqref{eq:mfmc_estimator} depends on the coefficients $a_1, \ldots, a_{L-1} \in \mathbb{R}$ and on the number of model evaluations $N_1, \ldots, N_L$. We define the costs of the MFMC estimator as
$
C(Y_{\text{MFMC}})=\sum_{l=1}^k N_l C_l,
$
since model $\model_l$ is evaluated $N_l$ times and each evaluation incurs costs $C_l$. The MFMC method selects $N_1, \ldots, N_L$, $a_1, \ldots, a_{L-1}$ such that the costs of the MFMC estimator are equal to a given computational budget $C(Y_{\text{MFMC}})=B$ and the mean square error of the MFMC estimator is minimal. 
With mild assumptions on the cost, the optimal $N_1, \ldots, N_L$, $a_1, \ldots, a_{L-1}$ can be calculated through functions of the cost and correlations coefficients between models~\cite{peherstorfer2016optimal}.

\subsection{Multi-level Multi-fidelity Monte Carlo}
Now, we present our proposed unified framework of multi-level multi-fidelity Monte Carlo (MLMF). 
Let us first consider a unified notation system for MLMC and MFMC. Let $P_l^{(N_l)}$ denote the estimate of $P$ at level $l$ with $N_l$ samples. This way, the MLMC estimator can be written as
\begin{equation}
\label{eq:mlmc_uni}
\begin{aligned}
    Y_{\text{MLMC}} & = N_1^{-1} \sum_{n=1}^{N_1}P_1^n + \sum_{l=2}^L \left\{N_l^{-1} \sum_{n=1}^{N_l}(P_l^n - P_{l-1}^n)\right\} \\
    & = P_1^{(N_1)} + \sum_{l=2}^L \left[P_l^{(N_l)} - P_{l-1}^{(N_l)}\right],
\end{aligned}
\end{equation}
and the MFMC estimator can be written as
\begin{equation}
\label{eq:mfmc_uni}
\begin{aligned}
    Y_{\text{MFMC}} & = P_L^{(N_L)} + \sum_{l=1}^{L-1} a_l \left[P_l^{(N_l)} - P_{l}^{(N_{l+1})}\right].
\end{aligned}
\end{equation}



It is clear to see that $Y_{\text{MLMC}}$ in~\eqref{eq:mlmc_uni} and $Y_{\text{MFMC}}$ in~\eqref{eq:mfmc_uni} share a highly similar structure. For MLMC, coupling in different time discretizations is leveraged through recycled noise realizations, and for MFMC, coupling in different model fidelities is leveraged through reused data samples, both to reduce estimation variance.
Based on these observations, we proposed the following MLMF estimator, to incorporate coupling in time and state spaces in one unified framework.
\begin{equation}
    \label{eq:mlmf}
    \begin{aligned}
    Y_{\text{MLMF}} & = P_L^{(N_L)} + \sum_{l=1}^{L-1} a_l \left[P_l^{(N_l)} - P_{l}^{(N_{l+1})}\right] \\
    & = a_1 P_1^{(N_1)} + \sum_{l=2}^L \left[a_l P_l^{(N_l)} - a_{l-1} P_{l-1}^{(N_l)}\right],
\end{aligned}
\end{equation}
where $a_l$ are parameters and $a_L = 1$. Note that $P_l^{(N_l)}$ is the estimate at level $l$ with possibly different time discretization $\Delta t_l$ similar to MLMC, and different surrogate models $\model_l$ similar to MFMC. 
Furthermore, we will show in Section~\ref{sec:mlmf_application} that this MLMF framework can take estimates from models of low-dimensional features instead of the original system states, to allow further flexibility and possible computation reduction. An illustration of the MLMF framework is shown in Fig.~\ref{fig:mlmf_diagram}

\setlength{\belowcaptionskip}{-4pt}
\begin{figure}[t]
    \centering
    \includegraphics[width=0.49\textwidth]{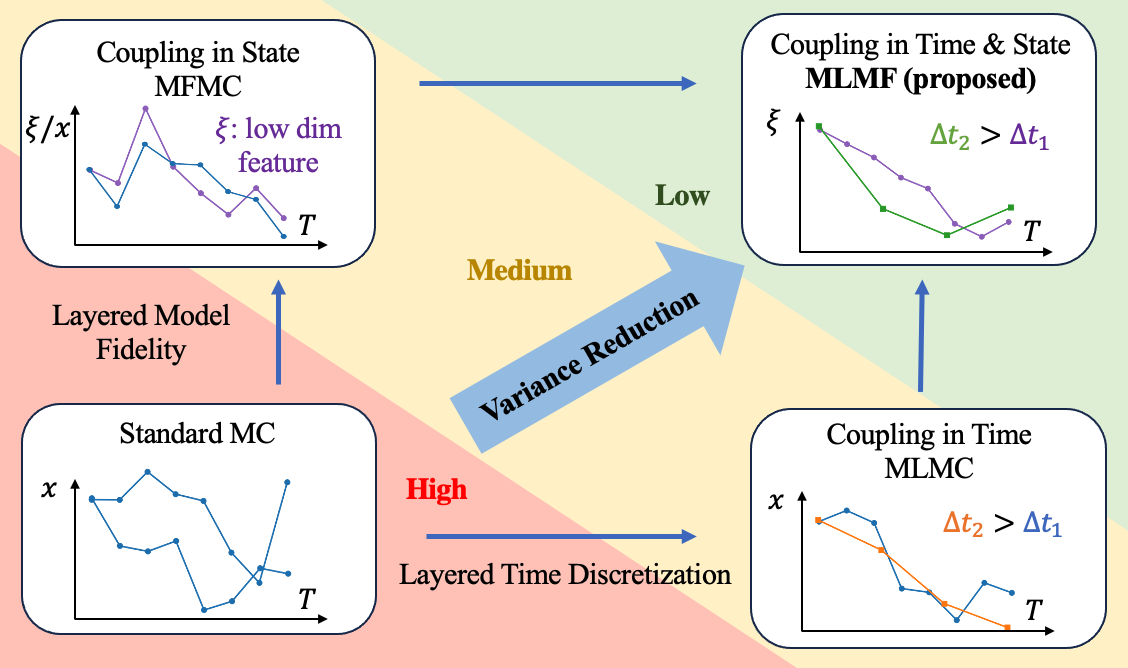}
    \caption{The proposed MLMF framework: coupling both in time and state space.}
    \label{fig:mlmf_diagram}
\end{figure}

\begin{remark}
    The proposed MLMF framework uses layering structure in both time and state space with possible low-dimensional representations, as opposed to MLMC where only different time discretizations are used~\cite{giles2015multilevel} and MFMC where only different system models of same dimensions are used~\cite{peherstorfer2018multifidelity}.
\end{remark}

\begin{remark}
    The parameters $a_L = 1$ will ensure unbiased MLMF estimator under certain conditions, as shown later in Theorem~\ref{thm:unbiased_estimation}. Different configurations of $a_l$ for $l = 1, \cdots L-1$ will result in different variances at each level, which affect the sample complexity of the estimator. Optimal $a_l$ are derived in Theorem~\ref{thm:optimal_coefficients} and the corresponding estimation variance is derived in Theorem~\ref{thm:optimal_variance}.
\end{remark}

\section{Theoretical Analysis}
\label{sec:theoretical_analysis}
In this section, we generalize key theoretical results from MLMC and MFMC literature to our proposed MLMF framework. Specifically, we adapt the control variate estimator, variance reduction, and sample allocation theorems from~\cite{peherstorfer2016optimal,geraci2017multifidelity} to our setting, where surrogate models at different levels may operate on different state spaces but originate from a common full-state representation. In addition, we show such analysis in both the settings where samples from different layers are coupled or not.


We make the following assumptions to facilitate our theoretical analysis.


\begin{assumption}
\label{asm:known_correlation}
    For each level $l$, the correlation coefficient $\rho_{l,L}$ between $P_l$ and $P_L$ is known, and we have 
    \begin{equation}
        |\rho_{L,L}| > |\rho_{L-1,L}| > \cdots > |\rho_{1,L}|.
    \end{equation}
\end{assumption}

\begin{assumption}
\label{asp:descending_cost}
    Evaluation costs $C_l$ for each level $l$ are positive and known, and we have
    \begin{equation}
        \frac{C_{l+1}}{C_l} > \frac{\rho_{l+1,L}^2 - \rho_{l,L}^2}{\rho_{l,L}^2 -  \rho_{l-1,L}^2},    
    \end{equation}
    for $l = 2, \cdots, L-1$.
\end{assumption}

In practice, Assumptions~\ref{asm:known_correlation} and~\ref{asp:descending_cost} can be satisfied by constructing meaningful low-level models $h_l$ that have decreasing correlations with the high-level model $h_L$, and have decreasing sampling cost. The correlation coefficients can be estimated from samples~\cite{giles2015multilevel}. 

In the following theorems and proofs, for $l$-th level estimate $P_l$, we let $V_l$ be the variance, $\sigma_l = \sqrt{V_l}$ be the standard deviation, and $C_l$ be the sampling cost.

\subsection{Unbiasedness of the MLMF Estimator}

We present the following theorem to show that the proposed MLMF estimator is unbiased.

\begin{theorem}
\label{thm:unbiased_estimation}
The MLMF estimator~\eqref{eq:mlmf} is unbiased, \ie $\mathbb{E}[Y_{\text{MLMF}}] = \mathbb{E}[h_L(x)]$.
\end{theorem}

\begin{proof}
By linearity of expectation,
\begin{equation}
\begin{aligned}
    \mathbb{E}[Y_{\text{MLMF}}] & = \mathbb{E}[P^{(N_L)}] + \sum_{l=1}^{L-1} a_l (\mathbb{E}[P_l^{(N_l)}] - \mathbb{E}[P_{l}^{(N_{l+1})}]) \\
    & = \mathbb{E}[h_L(x)].  
\end{aligned}
\end{equation}
\end{proof}

\subsection{Optimal Coefficients}
We present the following theorem to derive the optimal coefficients $a_l^*$ for variance reduction. 

\begin{theorem}
\label{thm:optimal_coefficients}
Assume Assumptions~\ref{asm:known_correlation} and~\ref{asp:descending_cost}  hold. The optimal weights $a_l^*$ minimizing the variance of $Y_{\text{MLMF}}$ are given by
\begin{equation}
\label{eq:optimal_coefficient}
    a_l^* = \rho_{l,L}\cdot\frac{\sigma_L}{\sigma_l},
\end{equation}
where $\rho_{l,L}$ is the Pearson correlation coefficient between $P_l$ and $P_L$, and $\sigma_l$ is the estimate standard deviation at level $l$.
\end{theorem}

\begin{proof}
Consider the MLMF estimator at fixed level $l\in\{1,\dots,L-1\}$ defined as follow:
\begin{equation}
    Y_{\text{MLMF}}^l := P_L^{(N_L)} + a_l \left[P_l^{(N_l)} - P_{l}^{(N_{l+1})}\right].
\end{equation}
Its variance is given by
\begin{equation}
\begin{aligned}
    \Var[Y_{\text{MLMF}}^l] = & \Var\bigl(P_L^{(N_L)}\bigr) 
    + a_l^2\,\Var\bigl(P_l^{(N_l)} - P_l^{(N_{l+1})}\bigr) \\
    & + 2a_l\,\Cov\bigl(P_L^{(N_L)},\, P_l^{(N_l)} - P_l^{(N_{l+1})}\bigr).
\end{aligned}
\end{equation}
From standard Monte Carlo properties~\cite{giles2015multilevel} we get
\begin{equation}
    \Var\bigl(P_L^{(N_L)}\bigr) \;=\; \frac{V_L}{N_L}, 
\end{equation}
\begin{equation}
    \Var\bigl(P_l^{(N_l)} - P_l^{(N_{l+1})}\bigr) = V_l\left(\frac{1}{N_{l+1}} - \frac{1}{N_l}\right),
\end{equation}
\begin{equation}
    \Cov \bigl(P_L^{(N_L)},\, P_l^{(N_l)} - P_l^{(N_{l+1})}\bigr) 
    = \sigma_l \sigma_L \Big(\frac{\rho_{l,L}}{N_l} - \frac{\rho_{l,L}}{N_{l+1}}\Big).
\end{equation}
Hence,
\begin{equation}
    \Var[Y_{\text{MLMF}}^l] = \frac{V_L}{N_L}
    + \left(\frac{1}{N_{l+1}} - \frac{1}{N_l}\right)
         \Bigl(a_l^2 V_l - 2a_l\, \sigma_l \sigma_L\,\rho_{l,L}\Bigr).
\end{equation}
Since $N_l \ge N_{l+1}$, the prefactor is nonnegative, and minimizing the quadratic term in $a_l$ yields
\begin{equation}
    a_l^* \;=\; \frac{\sigma_l \sigma_L\, \rho_{l,L}}{V_l}.
\end{equation}
Since $V_l=\sigma_l^2$, we have
\begin{equation}
    a_l^* \;=\; \rho_{l,L}\,\frac{\sigma_L}{\sigma_l}.
\end{equation}
\end{proof}

\subsection{Variance Reduction and Consistent Estimation}

We present the following theorem to derive the estimation variance under the optimal coefficients $a^*_l$, and show that the proposed MLMF estimator is consistent. We first define the covariance between different levels under $a^*_l$ as follow.
\begin{equation}
\label{eq:cross_variance}
\begin{aligned}
    V_{\text{cross}} := & \; -2 \sum_{l=1}^{L-1} a^*_l \Cov \bigl(P_L^{(N_L)},\, P_l^{(N_l)} - P_l^{(N_{l+1})}\bigr) \\
    = & \;  2 V_L \sum_{l=1}^{L-1}  \rho_{l,L}^2 \left(\frac{1}{N_{l}} - \frac{1}{N_{l+1}}\right),
\end{aligned}
\end{equation}
where $\rho_{l,L}$ is the Pearson correlation coefficient between $P_l$ and $P_L$.

\begin{theorem}
\label{thm:optimal_variance}
Assume Assumptions~\ref{asm:known_correlation} and~\ref{asp:descending_cost}  hold.
Using optimal coefficients $a_l^*$ in~\eqref{eq:optimal_coefficient}, the variance of the MLMF estimator is
\begin{equation}
\label{eq:optimal_variance}
    \Var^*[Y_{\text{MLMF}}] = \sum_{l=1}^L \frac{R_l}{N_l},
\end{equation}
where $R_l$ is the variance contribution of each level $l$ defined as follows.
\begin{equation}
\label{eq:variance_contribution}
    R_l := \begin{cases}
        V_L (\rho_{l, L}^2 - \rho_{l-1, L}^2),\text{ if coupled } P_l^{(N_l)}, P_l^{(N_{l+1})} \\
        V_L (\rho_{l, L}^2 + \rho_{l-1, L}^2),\text{ if independent }P_l^{(N_l)}, P_l^{(N_{l+1})}
    \end{cases}
\end{equation}
with $\rho_{0, L} = 0$.
\end{theorem}

\begin{proof}
From variance of sum, we know that the variance of the MLMF estimator can be written out as
\begin{equation}
\begin{aligned}
    & \;\text{Var}^*[Y_{\text{MLMF}}] \\
    = & \;\text{Var}\left[ P_L^{(N_L)} \right] + \sum_{l=1}^{L-1} \text{Var}[a_l^*(P_l^{(N_l)} - P_l^{(N_{l+1})})] - V_{\text{cross}},
\end{aligned}
\end{equation}
where $V_{\text{cross}}$ is the covariance between different levels defined in~\eqref{eq:cross_variance}.
Since $\text{Var}\left[P_L^{(N_L)} \right] = \frac{\sigma_L^2}{N_L}$, and from variance of sum we get
\begin{equation}
\begin{aligned}
    & \Var[a_l^*(P_l^{(N_l)} - P_l^{(N_{l+1})})] \\ = \; &(a_l^*)^2 \Var[P_l^{(N_l)}] +(a_l^*)^2 \Var[P_l^{(N_{l+1})}] \\
    & \quad- 2 (a_l^*)^2 \text{Cov}(P_l^{(N_l)}, P_l^{(N_{l+1})}),
\end{aligned}
\end{equation}
which gives
\begin{equation}
\label{eq:optimal_variance_nonsimplify}
\begin{aligned}
\text{Var}^*[Y_{\text{MLMF}}] &=
\frac{\sigma_L^2}{N_L} + \sum_{l=1}^{L-1} (a_l^*)^2 \Big[\frac{\sigma_l^2}{N_l} + \frac{\sigma_l^2}{N_{l+1}} \\
&\quad - 2 \Cov(P_l^{(N_l)}, P_l^{(N_{l+1})})\Big] - V_{\text{cross}},
\end{aligned}
\end{equation}
where $V_{\text{cross}}$ is defined in~\eqref{eq:cross_variance}.

Note that if samples are shared across the first $N_{l+1}$ evaluations for level $l$ as in~\cite{peherstorfer2016optimal}, 
\begin{equation}
\label{eq:covariance_coupled}
    \Cov(P_l^{(N_l)}, P_l^{(N_{l+1})}) = \frac{\sigma_l^2}{N_{l}},
\end{equation} 
otherwise if independent samples are used for $P_l^{(N_l)}$ and $P_l^{(N_{l+1})}$, 
\begin{equation}
\label{eq:covariance_noncoupled}
    \Cov(P_l^{(N_l)}, P_l^{(N_{l+1})}) = 0.
\end{equation}
These two cases correspond to coupling or not for sampling in MLMF. For the two cases, plugging~\eqref{eq:covariance_coupled} or~\eqref{eq:covariance_noncoupled} into~\eqref{eq:optimal_variance_nonsimplify}, and leveraging the definition of $R_l$ in~\eqref{eq:variance_contribution} yields~\eqref{eq:optimal_variance}.
\end{proof}

Note that without coupling, we have $R_l = V_L (\rho_{l, L}^2 + \rho_{l-1, L}^2)$ in~\eqref{eq:variance_contribution}, which results in higher variance compared to standard MC:
\begin{equation}
    \Var^*[Y_{\text{MLMF}}] = V_L \sum_{l=1}^L \frac{\rho_{l, L}^2 + \rho_{l-1, L}^2}{N_l} \geq \frac{V_L}{N_L} = \Var[Y_{\text{MC}}].
\end{equation}
This emphasizes that coupling is essential for variance reduction in the MLMF estimator.

\begin{corollary}
    The proposed MLMF estimator with optimal coefficients $a^*_l$ is consistent, \ie as $N_l \rightarrow \infty$ for all $l = 1, \cdots L$, $\Var^*[Y_{\text{MLMF}}] \rightarrow 0$.
\end{corollary}

\begin{proof}
    From Theorem~\ref{thm:optimal_variance} we know that under $a^*_l$, each term of $\Var^*[Y_{\text{MLMF}}]$ in~\eqref{eq:optimal_variance} will converge to $0$ as $N_l \rightarrow \infty$.
\end{proof}



\subsection{Optimal Sample Allocation}

We present the following theorem to derive the optimal sample allocation under cost constraints.

\begin{theorem}
\label{thm:optimal_sample_alloc}
Given a total computational budget $B$, the optimal number of samples $N_l$ that minimize variance under cost constraint $\sum_{l=1}^L N_l C_l \leq B$ are given by:
\begin{equation}
N_l^* = \frac{B}{\sum_{j=1}^L \sqrt{C_j R_j}} \cdot \sqrt{\frac{R_l}{C_l}},
\end{equation}
where $R_l$ is the variance contribution at level $l$ defined in~\eqref{eq:variance_contribution}.
\end{theorem}

\begin{proof}
From Theorem~\ref{thm:optimal_variance} we know that the optimal variance with regard to $a^*_l$ is given by~\eqref{eq:optimal_variance}. We define the Lagrangian as follows
\begin{equation}
    \mathcal{L} = \text{Var}^*[Y_{\text{MLMF}}] + \lambda \left( \sum_{l=1}^L N_l C_l - B \right).
\end{equation}
Then we set the partial derivatives to zero and get
\begin{equation}
\label{eq:set_deriv_zero}
    \frac{\partial \mathcal{L}}{\partial N_l} = -\frac{\partial \text{Var}^*[Y_{\text{MLMF}}]}{\partial N_l} + \lambda C_l = 0,
\end{equation}
where
\begin{equation}
    \frac{\partial \text{Var}^*[Y_{\text{MLMF}}]}{\partial N_l} = \frac{R_l}{N_l^2}.
\end{equation}
Solving~\eqref{eq:set_deriv_zero} yields
\begin{equation}
    N_l^2 = \frac{R_l}{\lambda C_l} \quad \Rightarrow \quad  N_l^* = \sqrt{\frac{R_l}{\lambda C_l}},
\end{equation}
and enforcing cost constraint gives
\begin{equation}
    \sum_l N_l^* C_l = \sum_l \sqrt{\frac{R_l C_l}{\lambda}} = B,
\end{equation}
which gives
\begin{equation}
    \lambda = \left( \frac{\sum_l \sqrt{R_l C_l}}{B} \right)^2.
\end{equation}
Thus,
\[
N_l^* = \frac{B}{\sum_{j=1}^L \sqrt{R_j C_j}} \cdot \sqrt{\frac{R_l}{C_l}}.
\]
\end{proof}

Now we derive the following corollary, to demonstrate that the proposed MLMF estimator has reduced sample complexity compared to standard MC which has cost $O(\epsilon^{-3})$~\cite{giles2015multilevel}.

\begin{corollary}
    The proposed MLMF estimator with optimal coefficients $a^*_l$ has optimal sample complexity $O(\epsilon^{-2})$ for any error tolerance $\epsilon>0$.
\end{corollary}

\begin{proof}
By Theorem~\ref{thm:unbiased_estimation}, the MLMF estimator is unbiased, so achieving mean square error $\epsilon^2$ is equivalent to achieving $\Var(Y_{\text{MLMF}})\le\epsilon^2$.
By Theorem~\ref{thm:optimal_variance} and Theorem~\ref{thm:optimal_sample_alloc}, we know that with the optimal sample allocation and optimal $a^*_l$ at each level, the optimal estimation variance is given by
\begin{equation}
    \sum_{l=1}^L \frac{R_l}{N_l^*} = \sum_{l=1}^L \frac{R_l {\sum_{j=1}^L \sqrt{R_j C_j}}}{B} \sqrt{\frac{C_l}{R_l}} = \frac{ (\sum_{l=1}^L \sqrt{R_l C_l})^2}{B}.
\end{equation}
Since $\sum_{l=1}^L \sqrt{R_l C_l}$ is a constant, the computation cost to achieve $\epsilon^2$ variance is on the order of $B \sim O(\epsilon^{-2})$.

\end{proof}

\section{Risk and Optimal Control Estimation}
\label{sec:mlmf_application}
In this section, we present how to estimate safety probability and path integral optimal control with the proposed MLMF framework.
Let $\xi = p(x) \in \mathbb{R}^k$ be a (lower-dimensional) feature of the system state $x$, where $p: \mathbb{R}^n \rightarrow \mathbb{R}^k$.
Assume $\model_L$ corresponds to the ground truth system dynamics and take the original state $x \in \mathbb{R}^n$ as input, and $\model_l$ for $l = 1, \cdots, L-1$ are $L-1$ correspond to lower-level models that take features $\xi^{(l)} \in \mathbb{R}^{n_{l}}$ as input. Note that usually $n_{l} \leq n$ for $l = 1, \cdots, L-1$ for efficient computation (\ie features are low-dimensional representations of the original system state). 
An example of this can be found in Sec.~\ref{sec:experiments} for the safety probability estimation experiments.

For the safety probability estimation problem, according to~\cite{wang2024physics}, we assume at each level $l$ we have access to a function $r_l: \mathbb{R}^{n_{l}} \rightarrow \mathbb{R}$ such that $\phi(x)$ can be approximately represented using the feature $\xi^{(l)}$ as
\begin{equation}
\label{eq:safety_feature_rep}
    \phi(x) \approx r_l(\xi^{(l)}).
\end{equation}
Then, at each level $l$, the safety probability can be estimated solely from the features $\xi^{(l)}$ via
\begin{equation}
\label{eq:feature_safety_prob}
\begin{aligned}
    P_l = \mathbb{E} \left[\mathds{1}(r_l(\xi^{(l)}_{t}) > 0, \forall t \in [0, T]) \right], \\
\end{aligned}
\end{equation}
where $\mathds{1}(\cdot)$ is the indicator function, and the expectation is taken over the dynamics of model $\model_l$ with a nominal controller $u^{(l)}_t$ for $t \in [0, T]$. Then, given the state-space model at each level, the safety probability can be estimated by MLMF via~\eqref{eq:mlmf}. 

Similarly, for the path integral optimal control problem, we assume at each level $l$ we have access to a function $r_l: \mathbb{R}^{n_{l}} \rightarrow \mathbb{R}$ such that the running cost can be approximately represented using the feature $\xi^{(l)}$ as
\begin{equation}
\label{eq:cost_feature_rep}
    \ell(x) = \Psi(x) \approx r_l(\xi^{(l)}).
\end{equation}
Here, the terminal cost $\Psi$ is assumed to take the same form as the running state cost $\ell$ without loss of generality. 
We then define the cost-to-go function in the feature space at each level as
\begin{equation}
    S^u_l(t) = \int_{t}^{T} (r_l(\xi^{(l)}_\tau) + \frac{1}{2}\|u_\tau\|^2) \; d\tau,
\end{equation}
Then, we can calculate the optimal control at each level via
\begin{equation}
\label{eq:feature_optimal_control}
\begin{aligned}
    & P_l := u^{(l)*}_t = u^{(l)}_t + \\
    & \lim _{s \rightarrow t} \frac{\mathbb{E}_{\xi^{(l)}, u^{(l)}}\left\{\exp \left\{-S^u_l(t)\right\} \int_t^s d w_\tau \mid \xi_t^{(l)}=p_l(x_t)\right\}}{(s-t) \mathbb{E}_{\xi^{(l)}, u^{(l)}_t}\left\{\exp \left\{-S^u_l(t)\right\} \mid \xi_t^{(l)}=p_l(x_t)\right\}},
\end{aligned}
\end{equation}
where $p_l: \mathbb{R}^n \rightarrow \mathbb{R}^{n_l}$ is the mapping from the original state to the feature at level $l$. Finally we can estimate the optimal control through MLMF via~\eqref{eq:mlmf} with $P_l$ in~\eqref{eq:feature_optimal_control}. 
Note that the optimal control found from lower levels can be used as the sampling controller at higher levels to reduce variance~\cite{thijssen2015path}.

The procedures of safety probability and path integral control estimation with MLMF is shown in Alg.~\ref{alg:mlmf_app}. 

\begin{remark}
    In our setting, the highest-fidelity model $h_L$ operates on the full system state $x \in \mathbb{R}^n$, while each lower-fidelity model $h_l$ for $l = 1, \ldots, L-1$ operates on a lower-dimensional feature $\xi^{(l)} = p_l(x) \in \mathbb{R}^{n_l}$ derived from $x$, where $n_l \leq n$. Thus, the models might not share the same input space, but they are coupled through their common dependence on the same high-dimensional state $x$ and are evaluated on shared trajectories, to enable variance reduction.
\end{remark}

\begin{algorithm}[t]
\caption{Safety Probability and Path Integral Control with MLMF}\label{alg:mlmf_app}
\begin{algorithmic}
    \State \textbf{Given:} $L, N_l, a_l, \Delta t_l, \xi^{(l)}, r_l$

    \For {$l \text{ in } 1:L$} \Comment{number of layers}

    \If {$l = 1$}
    \State Estimate~\eqref{eq:feature_safety_prob} or~\eqref{eq:feature_optimal_control} for $P_1^{(N_1)}$ with time discretization $\Delta t_1$ and feature $\xi^{(1)}$ using $N_1$ samples

    \Else

    \State Estimate~\eqref{eq:feature_safety_prob} or~\eqref{eq:feature_optimal_control} for $P_l^{(N_l)}$ and $P_{l-1}^{(N_l)}$ with time discretization $\Delta t_l$ and $\Delta t_{l-1}$ and feature $\xi^{(l)}$ and $\xi^{(l-1)}$ using $N_l$ samples
    
    \EndIf

    \EndFor

    \State \textbf{return} $a_1 P_1^{(N_1)} + \sum_{l=2}^L \left[a_l P_l^{(N_l)} - a_{l-1} P_{l-1}^{(N_l)}\right]$
    
\end{algorithmic}
\end{algorithm}

\section{Experiments}
\label{sec:experiments}
In this section, we present simulation results to show efficacy of the proposed method.
We first show in a safety probability estimation problem how the proposed MLMF method outperforms standard MC and subset simulation method~\cite{au2001estimation}, which is commonly used in rare event simulation. Then, we demonstrate that the proposed method can efficiently solve path integral optimal control problems.

\subsection{Safety Probability Estimation}

The closed-loop system dynamics is given by
\begin{equation}
    d \mathbf{x} = -K \cdot \mathbf{x} \; dt + \sigma \cdot dW_t,
\end{equation}
where $\mathbf{x} = [x_1, x_2, \cdots, x_5]^\top \in \mathbb{R}^5$ is the state, $W_t$ is the 5-dimensional standard Wiener process with $W_0 = \boldsymbol{0}$, $K = [K_1, K_2, K_3, K_4, K_5] = [1, 0.7, 0.6, 0.5, 0.4]$ is the system parameter 
and $\sigma = [\sigma_1, \sigma_2, \sigma_3, \sigma_4, \sigma_5] = [1, 0.7, 0.6, 0.5, 0.4]$ is the noise magnitude.
The safe set is given by $\mathcal{C} = \{\mathbf{x}: x_i \geq 0, \forall i = 1, 2, \cdots, 5\}$.
The computation cost for sampling each dimension $i \in \{1, 2, 3, 4, 5\}$ is assumed to be 
$C_1 = 1, \; C_2 = 10, \; C_3 = 20, \; C_4 = 50, \; C_5 = 500$.

We set the initial state to be $\mathbf{x}_0 = [1, 1, \cdots, 1]^\top$, and estimate the safety probability $F(\mathbf{x}_0)$ defined in~\eqref{eq:safe_prob_def} with horizon $T = 1$ using the proposed MLMF method with $a_l \equiv 1$, $\Delta t_l = 0.1$ for $l = 1, 2, 3$ and $\Delta t_l = 0.05$ for $l = 4, 5$. We obtain each level of estimate by simulating the first $l$ dimensions of the system (\ie $\xi^{(l)} = [x_1, \cdots, x_l]^\top$) and calculating the safety probability with regard to safe set $\mathcal{C}_l = \{\xi^{(l)}: x_i \geq 0, \forall i = 1, \cdots, l\}$.
We compare with \textbf{(i)} standard MC where we directly estimate the safety probability with the 5-dimensional dynamics, \textbf{(ii)} subset simulation where we multiply conditional probabilities of failure for transitional safe set $\mathcal{C}_p = \{\mathbf{x}: x_i \geq p, \forall i = 1, \cdots, 5\}$ with $p = \{0.8, 0.5, 0.3, 0.2, 0\}$, as well as \textbf{(iii)} MLMC (layering in time only) and \textbf{(iv)} MFMC (layering in state space only). Fig.~\ref{fig:risk_comparison} shows the absolute error of the estimate with regard to computation cost for all methods. The ground truth value is obtained by running standard MC with $10^{8}$ samples. We can see that with more samples, the absolute errors decrease for all methods. Layering in time (MLMC) or in state (MFMC) alone results in error reduction compared to standard MC and subset simulation. By leveraging both, the proposed MLMF outperforms all other methods.
\setlength{\belowcaptionskip}{-4pt}
\begin{figure}
    \centering
    \includegraphics[width=7cm]{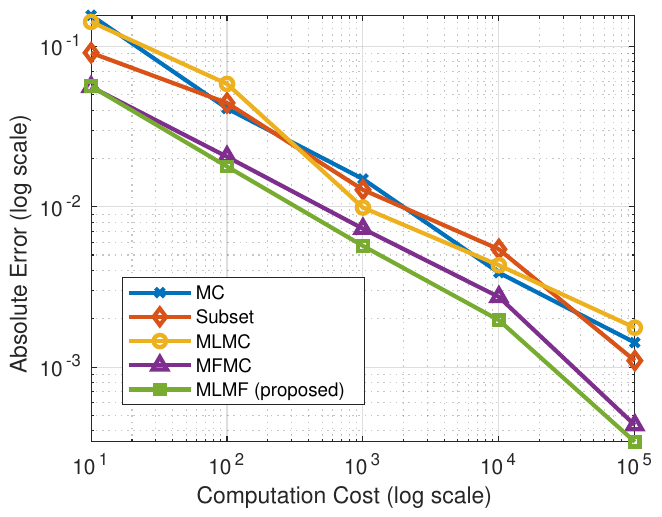}
    \caption{Absolute error of estimates for safety probability.}
    \label{fig:risk_comparison}
\end{figure}


\subsection{Path Integral Control}
We consider a system with discretized dynamics
\begin{equation}
    x_{k+1} = K x_k + (u_k \; dt + \sqrt{dt} \; w_k),
\end{equation}
where $x \in \mathbb{R}$ is the state, $u \in \mathbb{R}$ is the control input, $dt = 0.05$ is the time step, $w \sim \mathcal{N}(0, 1)$ is the noise sampled from a standard normal distribution, $K$ is the parameter of the system with ground truth $K_3 = 0.9$. We also have access to two surrogate models with $K_2 = 0.96$ and $K_1 = 0.99$ with the same form of the dynamics. The running cost for control is given by 
\begin{equation}
    g(x, u) = \frac{1}{2}\|x\|^2 + \frac{1}{2}\|u\|^2,
\end{equation}
and the terminal cost is set to 0 for simplicity. The computation cost for sampling each model is assumed to be $C_1 = 1, C_2 = 20, C_3 = 40$.

We set the initial state $x_0 = 10$, and run standard path integral control with three models as well as the path integral control with MLMF estimates for horizon $T = 5$. We set $a_1 = 0.05, a_2 = 0.1, a_3 = 1$ in the MLMF scheme. Fig.~\ref{fig:cost_computation} shows the optimal control cost and the computation cost. We can see that from model 1 to model 3 the control cost decreases but the computation cost increases, as sampling the more accurate model is more costly. The proposed MLMF method has similar optimal control cost as the ground truth model while having the least computation cost, indicating better performance vs. computation trade-offs. Fig.~\ref{fig:control_input} visualizes the optimal control inputs with different models, and we can see that the proposed MLMF method can closely recover the optimal control from the ground truth model even with much less computation.

\setlength{\belowcaptionskip}{0pt}
\begin{figure}
    \centering
    \includegraphics[width=0.48\textwidth]{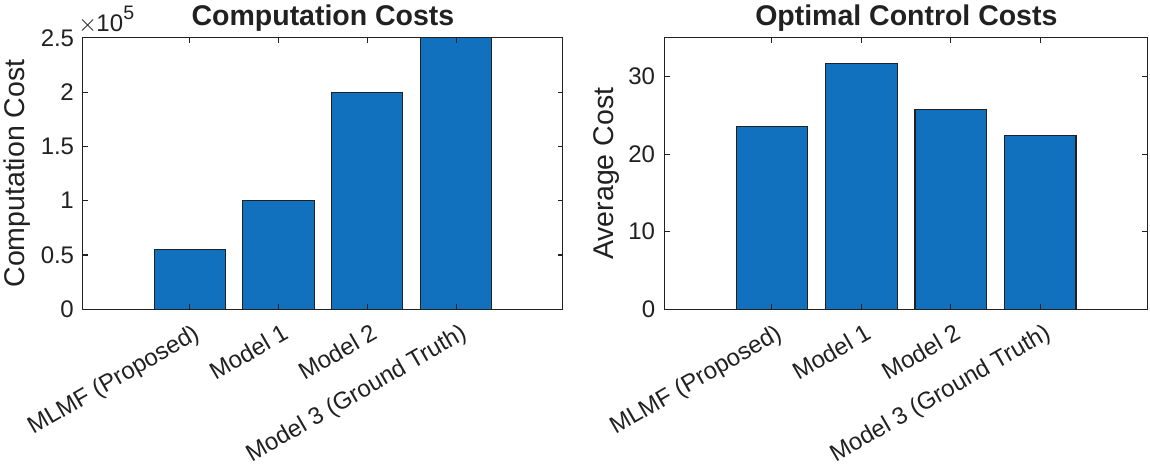}
    \caption{Optimal control cost (left) and computation cost (right) of the path integral control problem.}
    \label{fig:cost_computation}
\end{figure}

\begin{figure}
    \centering
    \includegraphics[width=5.6cm]{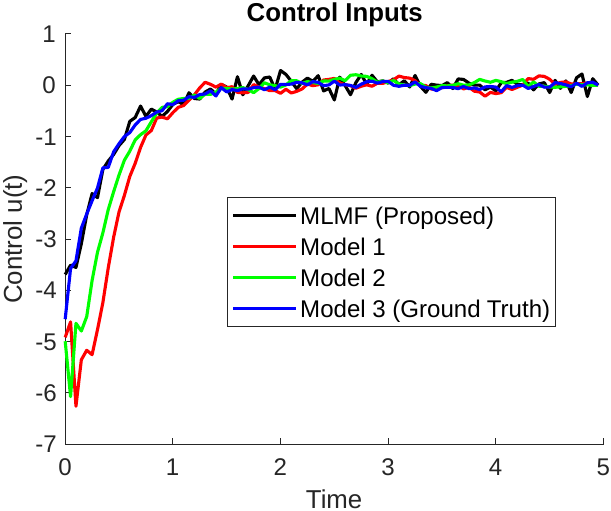}
    \caption{The optimal control input with different models.}
    \label{fig:control_input}
\end{figure}


\section{Conclusion}
\label{sec:conclusion}
In this paper, we consider estimating safety probability and path integral optimal control in stochastic systems. We propose a multi-level multi-fidelity Monte Carlo estimator that leverages layering structure in both time and state space. We show that under certain conditions on the sample number and parameter choice that the estimator is unbiased and consistent, and present sample complexity analysis. We show in experiments that the proposed framework has better computation and accuracy trade-offs than existing methods.

\bibliographystyle{unsrt}        
\bibliography{draft/bibliography/citation}

\end{document}